\documentclass[lettersize,journal]{IEEEtran}
\usepackage{amsmath,amsfonts}
\usepackage{algorithmic}
\usepackage{algorithm}
\usepackage{array}
\usepackage[caption=false,font=normalsize,labelfont=sf,textfont=sf]{subfig}
\usepackage{textcomp}
\usepackage{stfloats}
\usepackage{url}
\usepackage{verbatim}
\usepackage{graphicx}
\usepackage{cite}
\usepackage{balance}

\usepackage{array}
\usepackage{stfloats}
\usepackage{url}
\usepackage{graphicx}
\usepackage{textcomp}
\usepackage{xcolor}
\usepackage{tikz}
\usepackage{ellipsis}
\usetikzlibrary{calc}
\usetikzlibrary{decorations.pathreplacing,decorations.markings,shapes.geometric}
\usetikzlibrary{calc,patterns,angles,quotes,shapes,arrows.meta}
\usetikzlibrary{chains,spy}
\usepackage{cite}
\usepackage{amsmath,amssymb,amsfonts}
\usepackage[capitalise]{cleveref}
\usepackage{amsthm}
\newtheorem{lemma}{Lemma}
\newtheorem{corollary}{Corollary}

\usepackage[utf8]{inputenc}

\usepackage{mathtools}
\usepackage{bm}
\usepackage{etoolbox}
\usepackage{scalerel}

\usepackage{fancyhdr}

\newcommand{\va}{{\bm{a}}}

\newcommand{\vh}{{\bm{h}}}

\newcommand{\vn}{{\bm{n}}}

\newcommand{\vy}{{\bm{y}}}

\newcommand{\ma}{{\bm{A}}}

\newcommand{\mc}{{\bm{C}}}

\newcommand{\mw}{{\bm{W}}}

\newcommand{\my}{{\bm{Y}}}

\newcommand{\vrho}{{\bm{\rho}}}

\newcommand{\vdel}{{\bm{\delta}}}

\newcommand{\id}{{\mathrm{d}}}

\newcommand{\tr}{\mathrm{tr}}

\newcommand{\diag}{\mathrm{diag}}

\newcommand{\He}{\mathrm{H}}
\newcommand{\T}{\mathrm{T}}
\newcommand{\R}{\mathrm{R}}
\newcommand{\I}{\mathbf{I}}
\renewcommand{\j}{\mathrm{j}}
\newcommand{\NC}{\mathcal{N}_\mathbb{C}}
\newcommand{\ve}[1]{\mathrm{vec}\left(#1\right)}

\newcommand{\C}{\mathbb{C}}
\newcommand{\E}{\mathbb{E}}

\usepackage{amsthm}
\newtheorem{theorem}{Theorem}

\tikzstyle{block} = [draw, rectangle, 
minimum height=4em, minimum width=4em]
\tikzstyle{input} = [coordinate]
\tikzstyle{output} = [coordinate]
\tikzstyle{pinstyle} = [pin edge={to-,thin,black}]

\usetikzlibrary{positioning}

\tikzset{radiation/.style={{decorate,decoration={expanding waves,angle=90,segment length=5pt}}}}

\usetikzlibrary{spy}
\usepackage{pgfplots}
\usepackage{wrapfig}
\usetikzlibrary{arrows,shapes}
\usetikzlibrary{positioning,shapes.callouts}
\usepgfplotslibrary{groupplots,dateplot}
\usetikzlibrary{patterns,shapes.arrows}
\pgfplotsset{compat=newest}

\def\BibTeX{{\rm B\kern-.05em{\sc i\kern-.025em b}\kern-.08em
		T\kern-.1667em\lower.7ex\hbox{E}\kern-.125emX}}

\usepackage[acronym,shortcuts]{glossaries}
\newacronym{GMM}{GMM}{Gaussian mixture model}
\newacronym{PDF}{PDF}{probability density function}
\newacronym{MSE}{MSE}{mean square error}
\newacronym{CSI}{CSI}{channel state information}
\newacronym{CME}{CME}{conditional mean estimator}
\newacronym{ML}{ML}{maximum likelihood}
\newacronym{LS}{LS}{least squares}
\newacronym{LOS}{LoS}{line-of-sight}
\newacronym{NLOS}{NLoS}{non-\ac{LOS}}
\newacronym{DoA}{DoA}{direction-of-arrival}
\newacronym{DoD}{DoD}{direction-of-departure}
\newacronym{SNR}{SNR}{signal-to-noise ratio}
\newacronym{TNR}{TNR}{time-to-noise ratio}
\newacronym{BS}{BS}{base station}
\newacronym{JCAS}{JCAS}{joint communication and sensing}
\newacronym{MIMO}{MIMO}{multiple-input-multiple-output}
\newacronym{MMSE}{MMSE}{minimum \ac{MSE}}
\newacronym{NMSE}{NMSE}{normalized \ac{MSE}}
\newacronym{RMSE}{RMSE}{root \ac{MSE}}
\newacronym{LMMSE}{LMMSE}{linear minimum \ac{MSE}}
\newacronym{MT}{MT}{mobile terminal}
\newacronym{UE}{UE}{user equipment}
\newacronym{OMP}{OMP}{orthogonal matching pursuit}
\newacronym{CS}{CS}{compressed sensing}
\newacronym{ULA}{ULA}{uniform linear array}
\newacronym{DFT}{DFT}{Discrete Fourier Transform}
\newacronym{CRB}{CRB}{Cramer-Rao bound}
\newacronym{MUSIC}{MUSIC}{multiple signal classification}
\newacronym{rMUSIC}{rMUSIC}{root multiple signal classification}
\newacronym{ESPRIT}{ESPRIT}{estimation of signal parameters via rotational invariance techniques}
\newacronym{GE}{GE}{gridded estimator}
\newacronym{AWGN}{AWGN}{additive white Gaussian noise}
\newacronym{OFDM}{OFDM}{orthogonal frequency division multiplexing}
\newacronym{ADC}{ADC}{analog-to-digital converter}
\newacronym{wlog}{w.l.o.g.}{without loss of generality}
\newacronym{mmWave}{mmWave}{millimeter wave}
\newacronym{MAP}{MAP}{maximum a-posteriori}
\newacronym{PBCE}{PBCE}{parametric Bayesian channel estimator}
\newacronym{CGLM}{CGLM}{conditional Gaussian latent model}
\newacronym{AB}{AB}{asymptotic bound}

\newcommand{\change}[1]{\textcolor{black}{#1}}
\newcommand{\new}[1]{\textcolor{black}{#1}}

\usepackage{fancyhdr}

\fancypagestyle{cfooter}{ %
	\fancyhf{} % remove everything
	\cfoot{	\small This work has been submitted to the IEEE for possible publication. Copyright may be transferred without notice, after which this version may no longer be accessible.
	}
}

\begin{document}

\title{
On the Asymptotic MSE-Optimality of Parametric Bayesian Channel Estimation in mmWave Systems
}
\author{\IEEEauthorblockN{Franz Weißer,~\IEEEmembership{Graduate Student Member,~IEEE,} and Wolfgang Utschick,~\IEEEmembership{Fellow,~IEEE}
\thanks{This work was supported by the Federal Ministry of Research, Technology and Space of Germany in the programme of “Souverän. Digital. Vernetzt.”. Joint project 6G-life, project identification number: 16KISK002.}% <-this % stops a space
\thanks{The authors are with the TUM School of Computation, Information and Technology, Technical University of Munich, 80333 Munich, Germany (e-mail: franz.weisser@tum.de).}}}

\maketitle

\thispagestyle{cfooter}

\begin{abstract}
    The mean square error (MSE)-optimal estimator is known to be the conditional mean estimator (CME). 
    This paper introduces a parametric channel estimation technique based on Bayesian estimation.
    This technique uses the estimated channel parameters to parameterize the well-known LMMSE channel estimator.
    We first derive an asymptotic CME formulation that holds for a wide range of priors on the channel parameters.
    Based on this,
    we show that parametric Bayesian channel estimation is MSE-optimal for high signal-to-noise ratio (SNR) and/or long coherence intervals, i.e., many noisy observations provided within one coherence interval.
    Numerical simulations validate the derived formulations.

\end{abstract}


\section{Introduction}

\IEEEPARstart{A}{ccurate} channel estimation is a critical, long-standing issue in wireless communications systems.
From classical estimation theory, it is well-understood that the \ac{CME} is the optimal estimator in terms of \ac{MSE} and all other Bregman loss functions~\cite{Banerjee2005}.
Thus, many works have been dedicated to analyzing the \ac{CME} and approximations thereof, e.g.,~\cite{Guo2011,Neumann2018,Koller2022,Weisser2024b}.
As the \ac{CME} needs access to the true underlying distribution, this distribution is estimated (implicitly or explicitly) based on a representative data set using machine learning.

In contrast to stochastic models, parametric channel estimators leverage the channel's structural properties for estimation.
In recent years, several such estimators have been proposed for \ac{mmWave} systems, e.g.,~\cite{Yang2001,Shafin2016,Larsen2009,Hassan2020,Zhang2021},
\change{as the assumption of sparse channels is valid for high frequencies.}
These techniques \change{use unbiased estimators for the channel parameters and directly construct the channel accordingly.}
Additionally, parameter and parametric channel estimation become more relevant with the proposal of \ac{JCAS} systems, as other system functionalities either provide or need the parameter information.
A comprehensive overview of parametric channel estimation techniques for \ac{mmWave} systems can be found in~\cite{Hassan2020}.

Naturally, the question arises if parametric channel estimation can achieve the \ac{MMSE}.
It is shown in~\cite{Yang2001,Shafin2016} that \ac{DoA}-based channel estimation can achieve superior performance compared to \change{an approximate \ac{LMMSE} estimator, 
which leverages estimates of the second order statistics based on the observations.}
\change{The authors in~\cite{Larsen2009} derive bounds}
based on the \ac{CRB} of the error covariance matrix, which serve as a lower bound for unbiased parametric channel estimators.
As the \ac{CME} is not an unbiased estimator \change{and is conceptually capable of extending the maximum-likelihood estimation principle}, 
\change{solely focusing on the \ac{CRB} analysis is generally not optimal as the \ac{MMSE} may lie below the \ac{CRB}.}

\emph{Contributions:} 
We derive an expression of the \ac{CME} for high \ac{SNR} and/or long coherence intervals.
\change{We show that in this regime the \ac{CME} does not utilize} \change{prior distribution knowledge but solely works with the current observation.}
Based on the Bayesian approach, we introduce a parametric channel estimator that achieves this \ac{MMSE} in the asymptotic region, enabling \ac{MSE}-optimal channel estimates without knowledge of the underlying channel \change{or channel parameter} distribution.
Lastly, rigorous simulations show the validity of the derived bounds in the asymptotic region.

\section{Preliminaries}

\subsection{System Model}

\change{In time-varying \ac{MIMO} systems with multiple subcarriers, such as typical \ac{OFDM} systems, the baseband model of resource block $t$ can be described by}
\begin{align}
    \bm{\mathcal{Y}}\change{(t)} = \bm{\mathcal{H}}\change{(t)} + \bm{\mathcal{N}}\change{(t)} \in \C^{N_f\times N_t \times N_\R \times N_\T}, \label{eq:sys_model_tensor}
\end{align}
where $\bm{\mathcal{H}}\change{(t)}$ and $\bm{\mathcal{N}}\change{(t)}$ are tensors denoting the channel and the additive white Gaussian noise and $N_f$, $N_t$, $N_\R$, and $N_\T$ denote the number of subcarriers, symbols, receive and transmit antennas, respectively.
The system model given in~\eqref{eq:sys_model_tensor} can equivalently be described in its vectorized form as
\begin{align}
    \vy\change{(t)} = \ve{\bm{\mathcal{Y}}\change{(t)}} 
    = \vh\change{(t)} + \vn\change{(t)},
\end{align}
with $\vy\change{(t)}\in\C^{N_f N_t  N_\R  N_\T}$ \change{and $\vn(t)\sim\NC(\bm{0},\sigma^2\I)$.}
\change{The geometric channel model
consisting of $L$ paths is given as}
\begin{align}
    \vh({t}) = \sum_{\ell}^L \alpha_\ell({t}) \va_f(\tau_\ell) \otimes \va_t(\nu_\ell) \otimes \va_\R(\change{\theta_\ell}) \otimes \va_\T(\change{\phi_\ell}), \label{eq:gen_chan_model}
\end{align}
where $\otimes$ denotes the Kronecker product and $\alpha_\ell({t})$
, $\tau_\ell$, $\nu_\ell$, \change{$\theta_\ell$}, and \change{$\phi_\ell$} 
denote the complex path loss,
delay, Doppler-shift, \ac{DoA}, and \ac{DoD} of the $\ell$-th path, respectively.
We assume constant parameters within a coherence interval $T$, i.e., constant over \change{$T/N_t$ blocks}, where only $\alpha_\ell({t})$ depends \change{on the block index $t\in\{1,\dots,T/N_t\}$.
Further, the path gain $\alpha_\ell(t)$ is assumed to be invariant within one block.}
The generalized steering vector, \change{which samples equidistantly in each of the respective domains, is defined as}
\begin{align}
    \va_\xi(\omega^\xi_\ell) = \frac{1}{\sqrt{N_\xi}}[1,e^{-\j\omega^\xi_\ell},\dots,e^{-\j(N_\xi-1)\omega^\xi_\ell}]^\T,
\end{align}
where $\xi \in\{f,t,\mathrm{R},\mathrm{T}\}$ and $\omega^\xi_\ell$ depends on $\tau_\ell$, $\nu_\ell$, \change{$\theta_\ell$}, and \change{$\phi_\ell$}, 
respectively.
This work assumes $N_f=N_t=N_\T=1$, as the results can be straightforwardly extended.
Hence, we use 
$\va(\omega_\ell)=\va_\mathrm{R}(\pi\sin(\change{\theta_\ell}))$ 
as a shortened notation.
\change{Further}, we assume i.i.d. path gains with $\alpha_\ell(\change{n}) \sim \NC(0,\rho_\ell)$\change{\footnote{\change{For $N_t > 1$, the assumption of i.i.d. Gaussian path gains is an approximation as the steering vector $\va_t(\cdot)$ models the phase evolution over time.}}} and $\sum_{\ell=1}^L\rho_\ell=N_\R$, which leads to
the \ac{SNR} defined as \change{$\sigma^{-2}$.}

\subsection{Channel Estimation}
We want to estimate the channel $\vh(T)$ given all observations $\my=[\vy(1),\dots,\vy(T)]$, where we drop the index of $\vh(T)$ due to notational brevity.
\change{We consider the \ac{MSE} as our optimality criterion, which is most frequently used and} defined as
\begin{align}
    \text{MSE} = \E\left[\|\vh - \hat{\vh}\|^2_2\right],
\end{align}
where $\hat{\vh}$ denotes the estimate of $\vh$.

\section{Asymptotic Region Analysis}
\label{sec:high_snr}

\change{In our work, the \textit{asymptotic region} is the region of high \ac{SNR} and/or long coherence intervals $T$.}
For our derivations, we approximate the inner product between steering vectors, for small enough $\Delta\omega$, by~\cite{Shafin2016}
\begin{align}
    |\va^\He(\omega)\va(\omega + \Delta\omega)|^2 
    \approx 1 - \frac{N_\R^2 (\Delta\omega)^2}{12}, \label{eq:inner_2}
\end{align}
where the Taylor expansion is used, and any higher-order terms are neglected.

\subsection{MMSE Channel Estimation}
\label{sec:mmse}

The optimal solution in terms of \ac{MSE} is known to be the \ac{CME} defined as~\cite{Neumann2018}
\begin{align}
    \hat{\vh}_\text{CME}(\my) &= \E[\vh\mid\my] 
    = \E[\E[\vh\mid\my,\vdel]\mid\my]
    \label{eq:cglm}
\end{align}
\change{where the last equality follows from the law of total expectation for any latent variable $\vdel$.
We introduce \ac{wlog} $\vdel=[\bm{\omega},\vrho]$ resulting in $\vh\mid\vdel$ 
being conditionally Gaussian distributed 
with zero mean, cf.~\cite{Boeck2024}}.
\change{The inner expectation of~\eqref{eq:cglm} results in the LMMSE estimator
\begin{align}
 \E[\vh\mid\my,\vdel]=\mw_\vdel\vy =\mc_{\vh\mid\vdel} \left(\mc_{\vh\mid\vdel} + \mc_\vn\right)^{-1}\vy. \label{eq:condLMMSE}
\end{align}
}
The \ac{CME}, given $\my$, can be calculated as~\cite{Neumann2018}
\begin{multline}
    \change{\hat{\vh}_\text{CME}(\my) = \E[\mw_\vdel\mid\my]\vy = \mw_\text{CME}(\my)\vy }\\
    = \frac{\int \exp\left(\frac{T}{\sigma^2}\tr(\mw_\vdel \hat{\mc}_\vy) + T\log\left|\I - \mw_\vdel\right|\right)\mw_\vdel p(\vdel)\id\vdel}{\int \exp\left(\frac{T}{\sigma^2}\tr(\mw_\vdel \hat{\mc}_\vy) + T\log\left|\I - \mw_\vdel\right|\right)p(\vdel)\id\vdel}\vy, \label{eq:mmse_int}
\end{multline}
with $\hat{\mc}_\vy = \frac{1}{T} \my\my^\He$, \change{and $\mw_\vdel$ and $p(\vdel)$ denoting the conditional \ac{LMMSE} filter as given in~\eqref{eq:condLMMSE} and the prior distribution of $\vdel$.}
This \ac{CME} can generally not be evaluated as the underlying distribution \change{$p(\vdel)$} is unknown \change{and the integral is intractable.}

For our asymptotic analysis of the \ac{CME}, we introduce the following lemma.

\begin{lemma}\label{le:prior_multi}
For a prior $p(\vdel)$ that is continuous on its support, let $\mathbb{S}\subset\mathrm{supp}(p(\vdel))$ be a \change{connected} subset of its support.
Further let us denote with $\mathrm{diam}(\mathbb{S})=\mathrm{sup}\{\change{\|}\vdel_1-\vdel_2\change{\|} \mid \vdel_1,\vdel_2\in \mathbb{S}\}$ the diameter of the set $\mathbb{S}$.
A prior $p(\vdel)$ that satisfies
\begin{align}
    &\underset{\vdel \in  \mathbb{S} }{\mathrm{sup}}\left|\frac{\partial p(\vdel)}{\partial\vdel}\frac{1}{p(\vdel)}\right|
    \ll \frac{1}{\mathrm{diam}(\mathbb{S})} 
    \label{eq:prior_constraint} 
\end{align}
is approximately constant over the set $\mathbb{S}$.
\end{lemma}
\begin{proof}
    See Appendix A in the supplementary material.
\end{proof}

\subsubsection{Single-Path}

In the special case of $L=1$ the \change{latent is} $\delta=[\omega]$ resulting in the conditional \ac{LMMSE}-filter as
\begin{align}
    \mw_\delta = \frac{N_\R}{N_\R+\sigma^2} \va_\delta\va_\delta^\He,
\end{align}
where $\va_\delta = \va(\omega)$ leading to the following theorem.

\begin{theorem}\label{th:cme_single}
Let us denote with $\hat{\omega}$ the estimate of $\omega$ given by the Bartlett beamformer.
Further we assume that~\cref{le:prior_multi} holds for the interval $\mathbb{S}=[\hat{\omega}-k,\hat{\omega}+k]$, with $k$ such that $1-2\exp(-k^2/2C)\approx1$.
Then, \change{in the asymptotic region}
the \ac{CME} filter in~\eqref{eq:mmse_int} for the considered system with $L=1$ is
\begin{align}
    \mw_\text{CME}(\my) = \frac{1}{\sqrt{2\pi C}} \int \exp\left(-\frac{(\delta-\hat{\omega})^2}{2C}\right)\mw_\delta \mathrm{d}\delta, \label{eq:mmse_filter}
\end{align}    
with $C=\frac{6\sigma^2(N_\R+\sigma^2)}{T N_\R^3\Bar{\alpha}}$, $\Bar{\alpha}=\frac{1}{T}\sum_{n=1}^T|\alpha({t})|^2$.
\end{theorem}
\begin{proof}
    See Appendix B in the supplementary material.
\end{proof}
We see in~\cref{th:cme_single}
\change{that in the asymptotic region,}
\change{the \ac{CME} does not depend on the prior $p(\delta)$}
\change{but solely on the Bartlett beamformer estimate $\hat{\omega}$~\cite{Bartlett1950}.}
\change{Additionally, in the asymptotic region, i.e., $C\to0$, there exists a $k$ for every continuous prior such that \cref{le:prior_multi} is fulfilled.}

\subsubsection{Multiple Paths}

In the case of $L >1$, the parameter vector consists of $\vdel=[\bm{\omega},\vrho]=[\omega_1,\dots,\omega_L,\rho_1,\dots,\rho_L]$ and the integral in~\eqref{eq:mmse_int} becomes $2L$ dimensional.
The conditional channel covariance of $\vh\mid\vdel$ is given as
\begin{align}
    \mc_{\vh\mid\vdel} = \ma\mc_\vrho\ma^\He,
\end{align}
with $\ma = [\va(\omega_1),\dots,\va(\omega_L)]$ and $\mc_\vrho = \diag(\rho_1,\dots,\rho_L) \in \C^{L\times L}$.
For large enough arrays, we assume favorable propagation, \change{cf.~\cite[Chap. 2.5.2]{Bjoernson2017},}
rendering the \ac{LMMSE} filter to
\begin{align}
    \mw_\vdel &= \ma\mc_\vrho(\mc_\vrho + \sigma^2\I)^{-1}\ma^\He\\
    &= \sum_{\ell=1}^L \frac{\rho_\ell}{\rho_\ell + \sigma^2}\va(\omega_\ell)\va^\He(\omega_\ell), \label{eq:sum_filter}
\end{align}
which is a sum of $L$ filters along the different directions.

In order to make the high \ac{SNR} analysis feasible, we assume perfect knowledge of the individual path gains. 
This assumption can be motivated by the fact that empirical results show no significant dependency on the estimated path gains once a certain level of estimation quality is achieved.
\begin{corollary}\label{cor:cme_multi}
If~\cref{le:prior_multi} holds for $\mathbb{S}=[\hat{\omega}_1-k_1,\hat{\omega}_1+k_1]\times\dots\times[\hat{\omega}_L-k_L,\hat{\omega}_L+k_L]$,
    based on~\eqref{eq:sum_filter} and perfect path gain knowledge, 
    the \ac{CME} filter is
\begin{multline}
    \mw_\text{CME}(\my) = \\
    \sum_{\ell=1}^L\frac{1}{\sqrt{2\pi C_\ell}}\frac{\rho_\ell}{\rho_\ell + \sigma^2}   \int\exp\left(-\frac{(\delta-\hat{\omega}_\ell)^2}{2C_\ell}\right)\va_\delta\va^\He_\delta \mathrm{d}\delta, \label{eq:mmse_filter_mult}
\end{multline}
with $C_\ell=\frac{6\sigma^2(\rho_\ell+\sigma^2)}{T N_\R^2 \rho_\ell\Bar{\alpha}_\ell}$, $\Bar{\alpha}_\ell=\frac{1}{T}\sum_{n=1}^T|\alpha_\ell({t})|^2$ and $\hat{\omega}_\ell$ denoting the Bartlett estimates of $\omega_\ell$.
\end{corollary}
\begin{proof}
    Based on the decomposition in~\eqref{eq:sum_filter}, the result follows from~\cref{th:cme_single} for each direction separately.
\end{proof}
As in \cref{th:cme_single}, the prior $p(\vdel)$ does not appear in~\eqref{eq:mmse_filter_mult} \change{if the given conditions hold},
particularly also if the individual $\omega_\ell$ statistically depend on each other.

\subsection{Parametric Bayesian Channel Estimation}

\change{The \ac{CME} in \cref{th:cme_single} and \cref{cor:cme_multi} solely relies on the parameter estimates $\hat{\vdel}=[\hat{\bm{\omega}},\vrho]$.
Similarly, in parametric channel estimation, these parameters are estimated as an intermediate step.
Commonly, these estimates are used
to directly build a channel estimate, 
e.g.,~\cite{Zhang2021}.
On the contrary, the conditional Gaussian property of $\vh\mid\vdel$ can be utilized, as done by the \ac{CME}, to parameterize the \ac{LMMSE} as}
\begin{align}
    \hat{\vh}_\text{\change{PBCE}}(\my) = \change{\E[\vh\mid\my,\hat{\vdel},\new{\sigma^2}]} = \mw_{\hat{\vdel}\new{,\sigma^2}}\vy, \label{eq:filter_doa}
\end{align}
which we refer to as the \ac{PBCE}. 
\change{This estimator does not utilize any prior distribution of $\vdel$, and uses}
the  
maximum a posteriori
estimate of the filter $\mw_\vdel$ instead of solving the \ac{CME} integral.
\change{Independent of the used parameter estimator,}
the \ac{PBCE} is generally biased, 
\change{hence, a comparison to the \ac{CME} is more applicable than to the \ac{CRB}.}

\change{\begin{theorem} \label{th:quad_conv}
    The difference between the \acp{AB} of the \ac{CME} and the \ac{PBCE}  vanishes quadratically with the noise variance as
\begin{align}
    \lim_{\sigma^2\to0} \frac{\log\left(|\textsc{CME}^\textsc{AB} - \textsc{PBCE}^\textsc{AB}|\right)}{\log \sigma^2} = 2. 
    \label{eq:lim_expo}
\end{align}
\end{theorem}
\begin{proof}
    \begin{table*}[b]
\normalsize
\hrule
\begin{align}
    \change{\text{CME}^\text{AB}} \hspace{-1pt}
    &=\hspace{-1pt} N_\R \hspace{-1pt}-\hspace{-1pt}2\sum_{\ell=1}^L\left(\frac{\rho_\ell}{\rho_\ell+\sigma^2}\right)\left(\rho_\ell(1\hspace{-1pt} - \hspace{-1pt}2 B\Bar{C}_\ell)\hspace{-1pt}-\hspace{-1pt} 2B\;\mathrm{CRB}_{\change{\omega}}\right)\hspace{-1pt} + \hspace{-1pt}\sum_{\ell=1}^L\left(\frac{\rho_\ell}{\rho_\ell+\sigma^2}\right)^2\left(\rho_\ell(1 \hspace{-1pt}-\hspace{-1pt} 4 B\Bar{C}_\ell)\hspace{-1pt} - \hspace{-1pt}2B\;\mathrm{CRB}_{\change{\omega}}\hspace{-1pt} +\hspace{-1pt}\sigma^2\right) \label{eq:nmse_mmse_mult}
\end{align}
   \begin{align}
    \change{\text{PBCE}^\text{AB}}
    &= N_\R \hspace{-1pt}-\hspace{-1pt}2\sum_{\ell=1}^L\frac{\rho_\ell}{\rho_\ell+{\sigma^2}}\left(\rho_\ell\hspace{-1pt} - \hspace{-1pt}2B\;\mathrm{CRB}_{\change{\omega}}\right) \hspace{-1pt}+\hspace{-1pt} \sum_{\ell=1}^L\left(\frac{\rho_\ell}{\rho_\ell+{\sigma^2}}\right)^2\left(\rho_\ell \hspace{-1pt}- \hspace{-1pt}2B\;\mathrm{CRB}_{\change{\omega}} \hspace{-1pt}+ \hspace{-1pt}\sigma^2\right) \label{eq:nmse_doa_mult}
\end{align} 
\end{table*}%
The \ac{AB} 
of the \ac{CME} 
is shown in~\eqref{eq:nmse_mmse_mult} at the bottom of the page with $B=\frac{N_\R^2}{24}$ and $\Bar{C}_\ell=\E[C_\ell]$.
The calculations for the single-path case are sketched in Appendix C in the supplementary material, which can be easily extended to the multi-path case.
The \ac{AB}
of the \ac{PBCE} 
can be found similarly, resulting in~\eqref{eq:nmse_doa_mult}, also shown at the bottom of the page.
The limit in~\eqref{eq:lim_expo} follows from~\eqref{eq:nmse_mmse_mult} and~\eqref{eq:nmse_doa_mult}.
\end{proof}}

\change{In the case of $T\to \infty$ the bounds converge linearly to each other.
Both \ac{MSE} formulations depend on the parameter's \ac{CRB} as the \ac{CME} and the \ac{PBCE} use (or are assumed to use) an unbiased parameter estimator.}
If a parameter estimator is chosen \change{for the \ac{PBCE}}, which converges faster to the \ac{CRB} than the Bartlett beamformer, \change{cf.~\cite{Gazzah2011},} the \ac{PBCE} converges faster to its asymptote than the \ac{CME}.\footnote{One should note that for $L>1$ standard \ac{DoA} estimators are only asymptotically efficient for high $T$, not high \ac{SNR}.}

\new{In case of a noise variance mismatch $\hat{\sigma}^2 = (1+\varepsilon)\sigma^2$, with some $\varepsilon\in\mathbb{R}$, the \ac{AB} of the \ac{PBCE} still converges quadratically to the \ac{CME}, cf. Appendix D in the supplementary material.}

\section{Numerical Simulations}
\label{sec:simulations}

\begin{figure}
	\centering
	\includegraphics[]{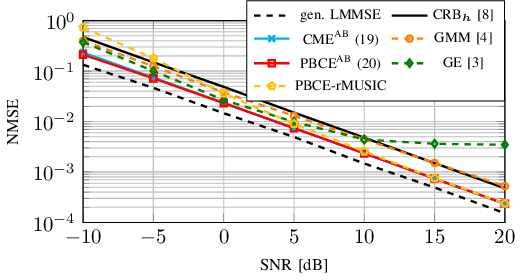}
	\\[-3mm]
	\caption{
		NMSE over the \ac{SNR} in the $L=1$ path and $N_\R=64$ antenna scenario based on $T=1$ observations.
	}
	\label{fig:nmse_SNR_1path}
\end{figure}

\change{We validate the derived bounds using}
the \ac{NMSE} defined as $\frac{1}{N_\R}\text{MSE}$,
\change{which we approximate}
by $M=10^3$ Monte-Carlo simulations.
If not stated otherwise, we use $N_\R=64$ antennas.
For the prior distribution, we use $4$ regions with probabilities $p_i\in\{0.1, 0.5, 0.2, 0.2\}$, where each region can be described by a Gaussian \ac{PDF} with $\mu_i\in\{-70^\circ,-30^\circ,20^\circ,60^\circ\}$ and $\sigma_i\in\{5^\circ,10^\circ,5^\circ,10^\circ\}$, respectively.
All path angles are drawn from the same 
region.\footnote{The quantitative results of the simulations are the same for any prior satisfying~\cref{le:prior_multi}.}
Additionally, we uniformly sample $\rho_\ell \sim \mathcal{U}[0,N_\R]$ and normalize it such that $\sum_{\ell=1}^L\rho_\ell=N_\R$.

For the \ac{PBCE},
we use \ac{rMUSIC}~\cite{Barabell1983} to estimate the $L$ \acp{DoA}, 
as \ac{rMUSIC} 
does not suffer from the grid mismatch problem and 
is shown to be asymptotically efficient.
We estimate the path gains with~\cite{Schmidt1986}
\begin{align}
    \hat{\vrho} = \diag\left((\hat{\ma}^\He\hat{\ma})^{-1}\hat{\ma}^\He(\hat{\mc}_\vy - \sigma^2\I)\hat{\ma}(\hat{\ma}^\He\hat{\ma})^{-1}\right), \label{eq:gain_est}
\end{align}
where $\hat{\ma}=[\va(\hat{\omega}_1),\dots,\va(\hat{\omega}_L)]$ holds the steering vectors based on the parameter estimates, \new{and assume knowledge of the true \ac{SNR} as the noise variance (or equivalently \ac{SNR}) can be estimated over a large scale interval ($\gg T$)}.

We also compare to the channel estimator from~\cite{Koller2022}, which does not use structural information (i.e., it is non-parametric) and approximates the \ac{CME} by learning a \ac{GMM} representation of the underlying channel distribution.
In the single-path case, we evaluate the \ac{GE} introduced in~\cite{Neumann2018}, which uses a sampled version of the true channel distribution, making the integral in~\eqref{eq:mmse_int} feasible to calculate.
Lastly, the utopian genie-aided \ac{LMMSE} estimator is simulated with
$
    \hat{\vh}_\text{gen. LMMSE}(\my,\vdel) = \E[\vh\mid\my,\vdel] = \mw_\vdel\vy,
$
\change{which uses perfect knowledge of the channel parameters.}

\subsection{Single-Path}

\cref{fig:nmse_SNR_1path} shows the channel estimation results based on $T=1$ observations for $L=1$ based on different \ac{SNR} values.
We can directly see that the \change{``CRB$_\vh$'' from~\cite{Larsen2009}}, which is the lower bound for unbiased channel estimators, \change{is outperformed by} the Bayesian-based estimators.
We see that the \ac{PBCE} based on \ac{rMUSIC} converges for high \ac{SNR} values to the derived \change{bounds} of 
the \ac{CME}~\eqref{eq:nmse_mmse_mult} and
\ac{PBCE}~\eqref{eq:nmse_doa_mult}.
This convergence validates the derived \acp{MSE} formulations.
Additionally, we observe that in the low \ac{SNR} region, the \change{\ac{PBCE}} degrades in performance compared to the \ac{GMM} and the \ac{GE}.
This decline is primarily because the \ac{PBCE} fully trusts its parameter estimates (which become worse for low \ac{SNR}) and does not take noise-dependent uncertainties into account.
\change{For higher \ac{SNR} values,} the \ac{GMM}-based \ac{CME} approximation shows inferior performance compared to \ac{PBCE}.
Furthermore, we see that for high \ac{SNR}, the \ac{GE} becomes victim to the grid mismatch problem, saturating for \ac{SNR} values above $10$dB.
Interestingly, the \ac{CME} does not coincide with the genie-aided \ac{LMMSE} even for high \ac{SNR}, which shows that \change{achieving} the \ac{MMSE}
\change{is possible without
perfect channel parameter knowledge.}

\subsection{Multiple Paths}

\begin{figure}
	\centering
	\includegraphics[]{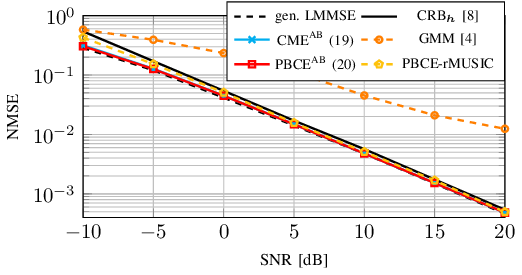}
	\\[-3mm]
	\caption{
NMSE over the \ac{SNR} for $L=3$ paths based on $T=16$ observations.
 }
	\label{fig:nmse_SNR_3path_16N}
\end{figure}

\begin{figure}
	\centering
	\includegraphics[]{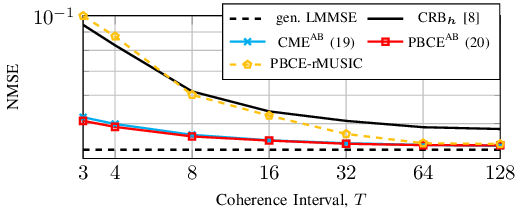}
	\\[-3mm]
	\caption{
    NMSE over the length of the coherence interval $T$ at SNR $=0$dB for $L=3$ paths.
 }
	\label{fig:nmse_snaps_3path}
\end{figure}

\begin{figure}
	\centering
	\includegraphics[]{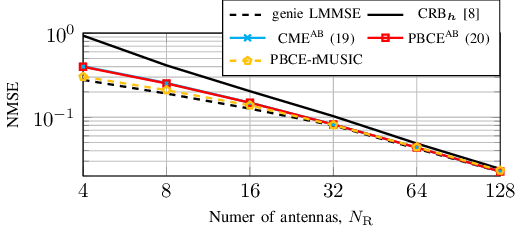}
	\\[-3mm]
	\caption{
    NMSE over the number of antennas at SNR $=0$dB and $T=64$ for $L=3$ paths.
 }
	\label{fig:nmse_ant_3path}
\end{figure}

For the remaining simulations, we assume channels with $L=3$ paths, where the number of paths is assumed to be known at the receiver.\footnote{This assumption is reasonable as good model order selection can be achieved for high \ac{SNR} and/or a high number of snapshots, cf.~\cite{Wax1985}.} In~\cref{fig:nmse_SNR_3path_16N}, the \ac{NMSE} performance is shown for $T=16$ observations
to ensure asymptotic performance of \ac{rMUSIC}.
Again, the two derived \ac{MSE} formulations in~\eqref{eq:nmse_mmse_mult} and~\eqref{eq:nmse_doa_mult} show the best \ac{NMSE}, which are only surpassed by the genie-aided \ac{LMMSE}.
One should note that in the multi-path setup, the derived \ac{MSE} formulation~\eqref{eq:nmse_mmse_mult} and~\eqref{eq:nmse_doa_mult} do not represent an asymptotic \ac{MSE} formulation but rather give a lower bound on the asymptotic \ac{MSE} as we assumed perfect path gain knowledge during the derivations.

As the asymptotic efficiency of \ac{rMUSIC} holds for large $T$,
~\cref{fig:nmse_snaps_3path} shows the performance for SNR$=0$dB over the length of the coherence interval $T$.
As expected, the performance decreases for all estimators approaching the \ac{CME} bound~\eqref{eq:nmse_mmse_mult}, validating the latter.

The last simulations in~\cref{fig:nmse_ant_3path} show the performance over different numbers of antennas in the asymptotic region at $T=64$ and \ac{SNR}$=0$dB. 
First, 
we realize that for low numbers of antennas, the \ac{PBCE} using \ac{rMUSIC} outperforms our derived bounds.
In this region, the assumption of favorable propagation no longer holds, making our bounds invalid.
Additionally, we see that for large $N_\R$, the gap between the \change{``CRB$_\vh$''} and the \ac{CME} formulation narrows.

\section{Conclusion}

In this work, we derived a \ac{CME} formulation for high \ac{SNR} and/or long coherence intervals in \ac{mmWave} systems, which holds for distributions on the channel parameters that fulfill a specific criterion on the slope of their \ac{PDF}.
Based on this formulation, we showed that \ac{PBCE} can achieve this \ac{MSE}-optimal performance in the asymptotic region.
Numerical simulations validate that \ac{PBCE} converges to the derived bounds, making it a superior channel estimation framework to unbiased parametric channel estimation.

\newpage
\balance
\bibliographystyle{IEEEtran}

\bibliography{IEEEabrv,mybib}

\newpage
% \section{Supplementary Material}
% \stepcounter{subsection}
\appendix

\begin{table*}[b]
\normalsize
\hrule
    \begin{align}
    \vh^\He\mw^\He\mw\vh %&= \frac{1}{2\pi C} \int\int \exp\left(-\frac{(t-\hat{\theta})^2}{2C}\right)\exp\left(-\frac{(s-\hat{\theta})^2}{2C}\right)\vh^\He\mw_t\mw_s^\He\vh \id t \id s \\
    = \frac{1}{2\pi C} \left(\frac{N_\R}{N_\R+\sigma^2}\right)^2 
    \int\int \exp\left(-\frac{(\delta-\hat{\omega})^2}{2C}\right)\exp\left(-\frac{(\zeta-\hat{\omega})^2}{2C}\right)\vh^\He\va_{\delta}\va_{\delta}^\He\va_{\zeta}\va_{\zeta}^\He\vh \id \delta \id \zeta. \tag{\ref{eq:phantom}}\label{eq:chan_pow_long_int}
\end{align}
\end{table*}

\subsection{Proof of~\cref{le:prior_multi}}
\label{app:proof_prior_multi}

If the constraint~\eqref{eq:prior_constraint} is fulfilled, 
based on the mean value theorem
% ~\cite[Theorem 5.10]{Rudin1976} 
% the constraint~\eqref{eq:prior_constraint} also holds 
for two arbitrary points $\vdel_1,\vdel_2 \in \mathbb{S}$ there exists a $\vdel \in \mathbb{S}\setminus \partial\mathbb{S}$ such that
\begin{align}
    \left|\frac{p(\vdel_1)-p(\vdel_2)}{(\vdel_1-\vdel_2)p(\vdel_2)}\right| = \left|\frac{\partial p(\vdel)}{\partial\vdel}\frac{1}{p(\vdel)}\right|\ll \frac{1}{\mathrm{diam}(\mathbb{S})}.
\end{align}
In the case of $p(\vdel_1)\geq p(\vdel_2)$ we have
% which can be reformulated as
\begin{align}
    1 \leq \frac{p(\vdel_1)}{p(\vdel_2)} \ll 1 + \frac{|\vdel_1-\vdel_2|}{\mathrm{diam}(\mathbb{S})} \leq 2,
\end{align}
and for $p(\vdel_1) < p(\vdel_2)$ we have
\begin{align}
    1 > \frac{p(\vdel_1)}{p(\vdel_2)} \gg 1 - \frac{|\vdel_1-\vdel_2|}{\mathrm{diam}(\mathbb{S})} \geq 0,
\end{align}
from which we can conclude $p(\vdel_1)\approx p(\vdel_2)\;\forall\vdel_1,\vdel_2 \in \mathbb{S}$.
% that
% \begin{align}
%     \frac{p(\vdel_1)}{p(\vdel_2)} \approx 1 \quad \forall\vdel_1,\vdel_2 \in \mathbb{S}.
% \end{align}
\qed

\subsection{Proof of~\cref{th:cme_single}}
\label{app:proof_cme_single}

The term $T\log\left|\I - \mw_\delta\right|$ in~\eqref{eq:mmse_int} is independent of $\delta$ and, hence, 
% cancels.
% Additionally, we assume a uniform prior on the angles $p(\theta)=1/\pi$.
the resulting optimal \change{\ac{CME} filter}, given $\my$, is
\begin{align}
    \mw_\text{CME}(\my) = \frac{\int \exp\left(\frac{TN_\R}{\sigma^2(N_\R+\sigma^2)}\va_\delta^\He\hat{\mc}_\vy \va_\delta\right)\mw_\delta p(\delta) \id\delta}{\int \exp\left(\frac{TN_\R}{\sigma^2(N_\R+\sigma^2)}\va_\delta^\He\hat{\mc}_\vy \va_\delta\right)p(\delta)\id\delta}. \label{eq:mmse_int_2}
\end{align}
% Let us further consider 
The argument of the exponential function \change{in~\eqref{eq:mmse_int_2}}
% , which 
resembles the Bartlett beamformer~\cite{Bartlett1950}.
\change{In the asymptotic region, the Bartlett spectrum exhibits one distinct peak, which we approximate} around its maximum value $\hat{\omega}$ \change{using~\eqref{eq:inner_2}} as
% \change{In the asymptotic region},
% we approximate the Bartlett spectrum around its maximum value $\hat{\omega}$ \change{using~\eqref{eq:inner_2}} as
\begin{align}
    \va_\delta^\He\hat{\mc}_\vy \va_\delta &\approx  \Bar{\alpha}\; |\va^\He_\delta\va(\hat{\omega})|^2 \approx \Bar{\alpha} \left(1 - \frac{ N_\R^2(\delta - \hat{\omega})^2}{12}\right),
\end{align}
which \change{approaches} equality for small $\delta - \hat{\omega}$.
% with $\Bar{\alpha}=\frac{1}{T}\sum_{t=1}^T\alpha(t)$.
The needed exponential function of~\eqref{eq:mmse_int_2} can therefore be written as
\begin{multline}
    \exp\left(\frac{TN_\R}{\sigma^2(N_\R+\sigma^2)}\va_\delta^\He\hat{\mc}_\vy \va_\delta \right) \\
    =\exp\left(\frac{TN_\R\Bar{\alpha}}{\sigma^2(N_\R+\sigma^2)}\right)\exp\left(-\frac{(\delta-\hat{\omega})^2}{2C}\right). 
\end{multline}
% with $C=\frac{6\sigma^2(M+\sigma^2)}{T\pi^2 M^3\Bar{\alpha}}$. 
%Based on this reformulation,
% the first factor can be removed in the numerator and denominator of \eqref{eq:mmse_int_2}.
Furthermore, 
we realize that the exponential function is proportional to a Gaussian \ac{PDF} with variance $C$.
Based on the Chernoff bound, the integral of the tails can be bounded as
\begin{align}
    \int_{\hat{\omega}+k}^\infty \exp\left(-\frac{(\delta-\hat{\omega})^2}{2C}\right) \id \delta \leq \sqrt{2\pi C} \exp\left(\frac{-k^2}{2C}\right),
\end{align}
respectively, which becomes negligible for large enough $k$.
% We further assume that the contributing support of this Gaussian \ac{PDF} is within $\pm4C$ around the mean, i.e., $\mathrm{supp}(\mathcal{N}(\delta;\hat{\omega},C))=[\hat{\omega}-4C,\hat{\omega}-4C]$.
If now the prior $p(\delta)$ satisfies~\cref{le:prior_multi} with $\mathbb{S}=[\hat{\omega}-k,\hat{\omega}+k]$, while $k$ being large enough to neglect the tails,
% its support includes 
% the estimate $\hat{\theta}$,
it is approximatively constant in the contributing interval,
% over $\mathrm{supp}(\mathcal{N}(\delta;\hat{\omega},C))$
and, hence, cancels \change{out}, which concludes the proof. \qed

\subsection{MSE Calculation for Single-Path}\label{app:mse_single}

The \ac{MSE} is given as
\begin{align}
    &\E[\|\vh - \hat{\vh}\|^2] = \E[\|\vh - \mw\vy\|^2] \nonumber\\
    &\quad= N_\R - \E[2\vh^\He\mw\vh] + \E[\|\mw\vh\|^2] + \E[\|\mw\vn\|^2],
\end{align}
where $\mw$ denotes the current filter of interest and $\hat{\vh}=\mw\vy$ is the estimate of $\vh$.
In the case of~\eqref{eq:mmse_filter} we have
% In the case of the high \ac{SNR} and/or long coherence interval \ac{CME} formulation we have
\begin{align}
    &\E[\vh^\He\mw\vh] \\
    &\quad=\E\left[\frac{1}{\sqrt{2\pi C}} \frac{N_\R}{N_\R+\sigma^2} \int \exp\left(-\frac{(\delta - \hat{\omega})^2}{2C}\right) |\vh^\He\va_\delta|^2 \id\delta\right] \\
    &\quad=\frac{N_\R}{N_\R+\sigma^2} \left(N_\R(1 - 2B\Bar{C}) - 2B\E[|\alpha|^2(\Delta\omega)^2]\right),
\end{align}
with $\Bar{C}=\E[C]$, where in the last step the approximation~\eqref{eq:inner_2} is used.
% Further, we have
% \begin{align}
%     \E[\vh^\He\mw\vh] = \frac{M}{M+\sigma^2} \left(M(1 - 2B\Bar{C}) - 2B\E[\alpha(\Delta\theta)^2]\right),
% \end{align}
% with $\Bar{C}=\E[C]$.
In~\cite{Gazzah2011} it has been proven that the Bartlett estimator is asymptotically unbiased and efficient in the single source case, hence, using
% where we can use 
the law of total expectation we have
% The used identity follows from the law of total expectation as
% Using the law of total expectation with further get
\begin{multline}
    \E[|\alpha|^2(\Delta\omega)^2] = \E_{\bm{\alpha}}\left[|\alpha|^2\;\E_{\left(\Delta\omega\right)^2\mid\bm{\alpha}}
    \left[\left(\Delta\omega\right)^2\big|\bm{\alpha}\right]\right]\\
    =  \frac{\sigma^2}{2T}\left[\Re\{\dot{\va}^\He(\I - \va\va^\He)^{-1}\dot{\va}\}\right]^{-1} =\mathrm{CRB}_{\change{\omega}},%\\
    %&= \E_{\bm{\alpha}}\left[\frac{\alpha}{\Bar{\alpha}}\mathrm{CRB}\right] %=\frac{\sigma^2M}{2T} K(\theta)\\ 
    %
\end{multline}
% \todo{condition on $|\alpha|^2$ and calculate CRB for that!}
% with%~\cite{Stoica1989}
% \begin{align}
%     \mathrm{CRB} = \frac{\sigma^2}{2T}\left[\Re\{\dot{\va}^\He(\I - \va\va^\He)^{-1}\dot{\va}\}\right]^{-1},-
%     % \mathrm{CRB} = \frac{\sigma^2}{2T}\left[\dot{\va}^\He\dot{\va}\right]^{-1},
% \end{align}
where $\dot{\va}=\partial\va(\omega)/\partial\omega$.
\refstepcounter{equation}% Increment the equation counter
\label{eq:phantom}%
% \begin{table*}[b]
% \normalsize
% \hrule
%     \begin{align}
%     \vh^\He\mw^\He\mw\vh %&= \frac{1}{2\pi C} \int\int \exp\left(-\frac{(t-\hat{\theta})^2}{2C}\right)\exp\left(-\frac{(s-\hat{\theta})^2}{2C}\right)\vh^\He\mw_t\mw_s^\He\vh \id t \id s \\
%     = \frac{1}{2\pi C} \left(\frac{M}{M+\sigma^2}\right)^2 
%     \int\int \exp\left(-\frac{(\delta-\hat{\theta})^2}{2C}\right)\exp\left(-\frac{(\xi-\hat{\theta})^2}{2C}\right)\vh^\He\va_{\delta}\va_{\delta}^\He\va_{\xi}\va_{\xi}^\He\vh \id \delta \id \xi. \label{eq:chan_pow_long_int}
% \end{align}
% \end{table*}

Similarily, for $\E[\|\mw\vh\|^2]$ two integrals have to be evaluated as shown in~\eqref{eq:chan_pow_long_int} at the bottom of this page, where the following approximation can be used
\begin{align}
    &\Re\{\vh^\He\va_{\delta}\va_{\delta}^\He\va_{\zeta}\va_{\zeta}^\He\vh\}\nonumber\\
    &\quad= \frac{|\alpha|^2}{N_\R^3}
    \frac{\sin\left(\frac{ N_\R(\delta-\omega)}{2}\right)}{\sin\left(\frac{\pi (\delta-\omega)}{2}\right)}
    \frac{\sin\left(\frac{ N_\R(\zeta-\delta)}{2}\right)}{\sin\left(\frac{\pi (\zeta-\delta)}{2}\right)}
    \frac{\sin\left(\frac{ N_\R(\omega-\zeta)}{2}\right)}{\sin\left(\frac{\pi (\omega-\zeta)}{2}\right)}\\
    % &\approx M \sinc\left(\frac{\pi M(t-\theta)}{2}\right)\sinc\left(\frac{\pi M(s-t)}{2}\right)\sinc\left(\frac{\pi M(\theta-s)}{2}\right) \\
    % &\approx M\left(1 - \frac{(\pi M(t-\theta))^2}{24}\right)\left(1 - \frac{(\pi M(s-t))^2}{24}\right)\left(1 - \frac{(\pi M(\theta-s))^2}{24}\right)\\
    % &\approx M \left(1 - \frac{(\pi M(t-\theta))^2}{24} - \frac{(\pi M(s-t))^2}{24} - \frac{(\pi M(\theta-s))^2}{24}\right)\\
    &\quad\approx |\alpha|^2\; \left(1 - B(\delta-\omega)^2 - B(\zeta-\delta)^2 - B(\omega-\zeta)^2\right),
\end{align}
which holds with equality for small differences in angles.
% The resulting overall \ac{MSE} is given in~\eqref{eq:nmse_mmse_mult}.

\subsection{\new{Noise Mismatch Analysis of the PBCE Bound}}
\label{app:sensi}
\new{Using the Taylor expansion, we can rewrite the \ac{PBCE} bound difference due to noise mismatch, i.e., $\hat{\sigma}^2=(1+\varepsilon)\sigma^2$, as
\begin{multline}
    \mathrm{PBCE}^\mathrm{AB}(\hat{\sigma}^2) - \mathrm{PBCE}^\mathrm{AB}({\sigma^2}) \\= \sum_{\ell=1}^L\frac{\rho_\ell^2(\varepsilon\sigma^2)^2}{(\rho_\ell + \sigma^2)^3} 
    + \mathcal{O}\left(\frac{(\varepsilon\sigma^2)^3}{(\rho_\ell + \sigma^2)^4}\right). \label{eq:pbce_est}
\end{multline}
In the high \ac{SNR} limit, this difference vanishes quadratically as
\begin{align}
     \lim_{\sigma^2\to0} \frac{\log\left(|\mathrm{PBCE}^\mathrm{AB}(\hat{\sigma}^2) - \mathrm{PBCE}^\mathrm{AB}({\sigma^2})|\right)}{\log \sigma^2} 
     % = \frac{2\sigma^2}{\hat{\sigma}^2-\sigma^2}
     = 2. 
\end{align}
Hence, the quadratic convergence of \ac{PBCE} is preserved as the error due to the noise mismatch also vanishes quadratically.}

\end{document}